\documentclass[10pt, pra, aps, twocolumn, showpacs, floatfix]{revtex4-2}
\usepackage{graphicx,subfigure,amsmath,amssymb,amsthm,hyperref,physics,lipsum,bm,bbm,bigints}
\newtheorem*{theorem*}{Theorem}
\graphicspath{{images/}}
\begin{document}

\title{Reverse engineering of one-qubit filter functions with dynamical invariants}
\author{R.~K.~L.~Colmenar}
\email{ralphkc1@umbc.edu}
\affiliation{Department of Physics, University of Maryland Baltimore County, Baltimore, MD 21250, USA}
\altaffiliation{Current address: Laboratory for Physical Sciences, University of Maryland, College Park, Maryland 20740, USA; ralphkc1@umd.edu}
\author{J.~P.~Kestner}
\affiliation{Department of Physics, University of Maryland Baltimore County, Baltimore, MD 21250, USA}

\begin{abstract}
We derive an integral expression for the filter-transfer function of an arbitrary one-qubit gate through the use of dynamical invariant theory and Hamiltonian reverse engineering. We use this result to define a cost function which can be efficiently optimized to produce one-qubit control pulses that are robust against specified frequency bands of the noise power spectral density. We demonstrate the utility of our result by generating optimal control pulses that are designed to suppress broadband detuning and pulse amplitude noise. We report an order of magnitude improvement in gate fidelity in comparison with known composite pulse sequences. More broadly, we also use the same theoretical framework to prove the robustness of nonadiabatic geometric quantum gates under specific error models and control constraints.
\end{abstract}

\maketitle

\section{Introduction}
Accurate manipulation of noisy quantum systems is an important problem in optimal control theory with potential applications in the field of chemical reaction control~\cite{Gordon_1997,Assion_1998,Rice_2002}, quantum sensing~\cite{Poggiali_2018,Muller_2018}, and quantum information processing (QIP)~\cite{Nielsen_Chuang} to name a few. In QIP, a typical strategy for suppressing errors due to noise is to use dynamical decoupling~\cite{Viola_1999,Khodjasteh_2005,Biercuk_2009,deLange_2010,Naydenov_2011,Bylander_2011} and composite pulse sequences~\cite{Levitt_1986,Brown_2004,Vandersypen_2005,Wang_2012,Kestner_2013,CalderonVargas_2017}. These techniques are designed to perturbatively suppress noise with correlation time scales that are much longer than the target evolution time (quasistatic noise). In many instances, however, quantum devices also suffer from non-static noise that fluctuates on the order of the evolution time or faster~\cite{Simmonds_2004,Bylander_2011,Dial_2013,Yoneda_2018}. Composite pulses have limited efficacy in such cases~\cite{Kabytayev_2014} and can even be detrimental to the quality of the generated quantum gate~\cite{Gungordu_2018}. 

An alternative solution to these control problems is to use pulse shaping techniques~\cite{Khaneja_2005,Daems_2013,Guo_2018,Gungordu_2019,Nobauer_2015,Barnes_2015,Motzoi_2009,Chen_2011}. The main idea of this approach is to find, either analytically or numerically, an appropriate set of time-dependent control Hamiltonian parameters that produces a desired evolution. Since the time-dependent Schr\"{o}dinger equation (TDSE) is generally not analytically tractable, analytical solutions are typically limited to simple pulse shapes \cite{Torosov_2011} or in restricted settings (e.g., for static error~\cite{Barnes_2015,Gungordu_2019} or state transfer protocols~\cite{Daems_2013}). Numerical solutions offer much more flexibility in the control landscape. When combined with the formalism of filter functions~\cite{Green_2013}, which characterizes the sensitivity of a control protocol to the power spectral density of the noise, it is possible to generate quantum gates that are robust against a specified spectral region of noise. Specifically, robust quantum gates are obtained by minimizing the overlap between the control's filter function and the noise power spectral density (PSD) in frequency space. This can be used, along with any control field constraints, to define a cost function to be minimized using, for example, gradient-based methods. Optimization algorithms that are designed for deep learning and are implemented in platforms such as TensorFlow~\cite{Tensorflow_2015} or Julia's Flux package~\cite{Flux_2018} are especially well-suited for these tasks owing to their built-in automatic differentiation capability. The power and flexibility offered by deep neural networks for solving quantum control problems has been demonstrated in a variety of recent works~\cite{Ball_2021,Bentley_2020,Baum_2021,Carvalho_2021,Gungordu_2020,Kanaar_2021}. However, filter function engineering typically involves solving the TDSE for the time evolution operator. It is possible to circumvent this, for example, using Hamiltonian reverse engineering based on the theory of dynamical invariants~\cite{Chen_2011}. Thus, it is possible to further reduce the computational workload of the optimization framework by reparameterizing the cost function in terms of dynamical invariant parameters.

In this work, we use dynamical invariant theory and Hamiltonian reverse engineering to derive an integral expression for the filter function of an arbitrary one-qubit gate and explore its theoretical and practical applications. Our work is structured as follows. We begin Sec.~\ref{sec:dynamical-invariants} by reviewing the theory of dynamical invariants. We follow this up with a derivation of the one-qubit filter function for an arbitrary noise model in terms of the dynamical invariant parameters. We explore the practical applications of our results in Sec.~\ref{sec:results} by numerically searching for optimal control solutions using deep neural networks. Specifically, we consider noise models with a $1/f$ noise spectrum~\cite{Dutta_1981} which is prevalent in solid-state qubits~\cite{Dial_2013,Yoneda_2018,Chan_2018,VanHarlingen_2004,Yoshihara_2006,Kim_2015}. In addition, we discuss in Sec.~\ref{sec:robustness-of-gp} some theoretical implications of our result by proving the robustness of geometric quantum gates against certain noise models under a strict only two-axis driving constraint. We then conclude and summarize our findings in Sec.~\ref{sec:conclusions}.

\section{Dynamical invariants}
\label{sec:dynamical-invariants}
We consider as our starting point a general one-qubit control Hamiltonian with three-axis driving,
\begin{equation}
\label{eq:hamiltonian}
H_c(t) = \frac{1}{2}
	\begin{bmatrix}
		\Delta(t) & \Omega(t) \mathrm{e}^{-i \varphi(t)} \\
		\Omega(t) \mathrm{e}^{i \varphi(t)} & -\Delta(t)
	\end{bmatrix}.
\end{equation}
This particular form is relevant in systems such as superconducting qubits~\cite{Koch_2007}, quantum dot spin qubits~\cite{Laucht_2017}, and NMR qubits~\cite{Gershenfeld_1997} to name a few, corresponding to the rotating wave approximation for a two-level system that is driven by an oscillating field with amplitude $\Omega$ at a carrier frequency detuned from resonance by $\Delta$, and with phase $\varphi$. Here three-axis driving means that all three control fields can be variably tuned to produce arbitrary Bloch sphere rotations. The solution to the time-dependent Schr\"{o}dinger equation with this Hamiltonian is not analytically tractable in general. It is possible, however, to use the theory of dynamical invariants to reformulate this problem so as to specify a resulting unitary evolution and then analytically calculate a time-dependent Hamiltonian that would produce it~\cite{Chen_2011}. A dynamical invariant $I(t)$ is a solution to the Liouville-von Neumann equation \cite{Lewis_1967}
\begin{equation}
\label{eq:liouville}
i \frac{\partial I(t)}{\partial t} - \left[H_c(t),I(t)\right] = 0.
\end{equation}
The eigenvectors $\ket{\phi_n(t)}$ of $I(t)$ are related to the solutions of the Schr\"{o}dinger equation by a global phase factor: $\ket{\psi_n(t)} = \mathrm{e}^{i\alpha_n(t)}\ket{\phi_n(t)}$, where $\alpha_n(t)$ are the Lewis-Riesenfeld phases given by~\cite{Lewis_1969}
\begin{equation}
\label{eq:lewis-riesenfeld-phase}
\alpha_n(t) = \int_{0}^{t} \ev**{i\frac{\partial}{\partial s}-H_c(s)}{\phi_{n}(s)}\,\mathrm{d}s.
\end{equation}
Within this framework, the time evolution operator $U_{c}(t)$ can be expressed as
\begin{equation}
\label{eq:time-evolution}
U_c(t) = \sum_{n = \pm} \mathrm{e}^{i \alpha_n(t)} \ket{\phi_n(t)}\bra{\phi_n(0)}.
\end{equation}
Thus, the theory of dynamical invariants effectively transforms the problem of solving the time-dependent Schr\"{o}dinger equation to finding an appropriate $I(t)$ that satisfies Eq.~\eqref{eq:liouville}. As a consequence, we are free to choose a parametrization for $U_{c}(t)$ by choosing the $\ket{\phi_n(t)}$ appropriately. Suppose that we choose
\begin{gather}
\ket{\phi_{+}(t)} = \cos\left[\frac{\gamma(t)}{2}\right] \mathrm{e}^{-i\beta(t)}\ket{0} + \sin\left[\frac{\gamma(t)}{2}\right]\ket{1},
\label{eq:eigenvectors-1}\\
\ket{\phi_{-}(t)} = \sin\left[\frac{\gamma(t)}{2}\right]\ket{0} - \cos\left[\frac{\gamma(t)}{2}\right]\mathrm{e}^{i\beta(t)}\ket{1},
\label{eq:eigenvectors-2}
\end{gather}
where $\gamma$ and $\beta$ are the dynamical invariant parameters, $I(t)\ket{\phi_n(t)} = \pm\Omega_0/2\ket{\phi_n(t)}$, and $\Omega_0$ is an arbitrary constant with units of frequency. This allows us to express $I(t)$ in a form similar to Eq.~\eqref{eq:hamiltonian}
\begin{equation}
\label{eq:dynamical-invariant}
I(t) = \frac{\Omega_0}{2}
	\begin{bmatrix}
		\cos(\gamma) & \sin(\gamma)\mathrm{e}^{-i \beta} \\
		\sin(\gamma)\mathrm{e}^{i \beta} & -\cos(\gamma)
	\end{bmatrix}.
\end{equation}
If we require Equations \eqref{eq:hamiltonian} and \eqref{eq:dynamical-invariant} to satisfy Eq.~\eqref{eq:liouville}, we are left with two coupled auxiliary equations \cite{Chen_2011}
\begin{gather}
\dot{\gamma} = -\Omega \sin(\beta-\varphi)
\label{eq:auxiliary-equations-1}\\
\Delta -\dot{\beta} = \Omega \cot(\gamma)\cos (\beta - \varphi),
\label{eq:auxiliary-equations-2}
\end{gather}
which, along with the appropriate boundary conditions, can be used to determine the control parameters $\Omega(t)$, $\Delta(t)$, and $\varphi(t)$ that targets a desired $U_c(t)$. This choice of parameterization allows us to write $U_{c}(t)$ strictly in terms of the dynamical invariant parameters and the Lewis-Riesenfeld phase:
\begin{equation}
\begin{aligned}
\label{eq:time-evolution-rewrite}
U_c(t) &= \mathrm{e}^{-i\frac{\beta(t)}{2}\sigma_{Z}} \mathrm{e}^{-i\frac{\gamma(t)}{2}\sigma_{Y}} \mathrm{e}^{i\frac{\zeta(t) - \zeta(0)}{2}\sigma_{Z}}\\
&\qquad \times \mathrm{e}^{i\frac{\gamma(0)}{2}\sigma_{Y}} \mathrm{e}^{i\frac{\beta(0)}{2}\sigma_{Z}},
\end{aligned}
\end{equation}
where $\alpha = \alpha_+ = -\alpha_-$ and we introduce a new dynamical invariant parameter
\begin{align}
    \zeta(t) &= 2\alpha(t) -\beta(t) \nonumber\\
    \label{eq:zeta}
    &= -\beta(0) + \int_{0}^{t} \frac{\dot{\gamma}\cot(\beta-\varphi)}{\sin\gamma}\,\mathrm{d}t'
\end{align}

The auxiliary equations provide a family of control solutions that allow us to reverse engineer a desired quantum gate. Since the gate only depends on the boundary values of the dynamical invariant parameters, there are infinitely many ways to generate the gate. It is desirable to use this freedom in the control Hamiltonian such that the resulting evolution is also robust against noise. To this end, filter functions provide a convenient method of quantifying the gate fidelity's susceptibility to noise with respect to its spectral properties~\cite{Green_2013}. 
The total one-qubit Hamiltonian in the presence of noise can be written as 
\begin{equation}
H(t) = H_c(t) + H_e(t),
\end{equation} 
where $H_c(t)$ is the ideal deterministic control Hamiltonian and $H_e(t)$ is the stochastic error Hamiltonian. More explicitly, $H_e(t)$ can generally be expressed as
\begin{equation}
\label{eq:error-H}
H_e(t) = \sum_{i=1}^3 \sum_{q} \delta_q(t)\chi_{q,i}(t)\sigma_i,
\end{equation}
where $q$ indexes a set of uncorrelated stochastic variables $\delta_q(t)$, $\chi_{q,i}(t)$ contains the sensitivity of the control parameters (which generally can be a function of the parameters themselves) to $\delta_q(t)$, and $\sigma_i$ are Pauli operators. For sufficiently weak noise, the average gate infidelity $\expval{\mathcal{I}}$ of the noisy evolution $U(t)$, which satisfies $i \dot{U}(t) = H(t) U(t)$ where $U(0)=\mathbbm{1}$, can be computed perturbatively. Up to the first-order Magnus expansion, we can compactly express the gate infidelity as (see Appendix \ref{app:derivation})
\begin{equation}
\label{eq:avg-infidelity}
\expval{\mathcal{I}} \approx \frac{1}{2\pi} \sum_{q} \int_{-\infty}^{\infty} S_{q}(\omega) F_{q}(\omega)\,\mathrm{d}\omega,
\end{equation}
where $S_q(\omega)$ denotes the noise PSD for the stochastic variable $\delta_q(t)$ and $F_{q}(\omega)$ is the corresponding first-order filter function which can be calculated using the following equations:
\begin{gather}
\label{eq:filter-function}
F_{q}(\omega) = \sum_k \left| R_{q,k}(\omega)\right|^2, \\
\label{eq:R-freq-dom}
R_{q,k}(\omega) = \sum_i\int_{0}^{T} \chi_{q,i}(t) R_{ik}(t) \mathrm{e}^{i\omega t}\,\mathrm{d}t, \\
\label{eq:adjoint-rep}
R_{ik}(t) = \frac{1}{2}\tr\left(U_{c}^{\dagger}(t)\sigma_{i}U_{c}(t)\sigma_{k}\right),
\end{gather}
where $T$ is the gate time. 

Combining Equations \eqref{eq:time-evolution} and \eqref{eq:adjoint-rep} allows us to express Eq.~\eqref{eq:R-freq-dom} as
\begin{equation}
\begin{aligned}
& R_{q,k}(\omega) = \frac{1}{2} \sum_{i,n,n'} \mel**{\phi_n(0)}{\sigma_k}{\phi_{n'}(0)}\\
& \times \int_{0}^{T} \mathrm{e}^{i(\alpha_n(t)-\alpha_{n'}(t)+\omega t)} \chi_{q,i}(t) \mel**{\phi_{n'}(t)}{\sigma_i}{\phi_n(t)}\,\mathrm{d}t.
\end{aligned}
\end{equation}
Thus, the filter function corresponding to $\delta_q(t)$ is given by
\begin{equation}
\begin{aligned}
\label{eq:filter-function-intermediate}
&F_{q}(\omega) = \sum_k R_{q,k}(\omega) R_{q,k}^{*}(\omega)\\
&= \frac{1}{4} \int_{0}^{T} \int_{0}^{T} \sum_{\substack{i,j,k,\\n,m,n^\prime , m^\prime}} \mel**{\phi_n(0)}{\sigma_k}{\phi_{n^{'}}(0)} \mel**{\phi_m(0)}{\sigma_k}{\phi_{m^{'}}(0)}\\
& \mathrm{e}^{i(\alpha_n(t_1)-\alpha_{n'}(t_1)+\omega t_1)} \chi_{q,i}(t_1) \mel**{\phi_{n'}(t_1)}{\sigma_i}{\phi_n(t_1)}\,\mathrm{d}t_1\\
&\mathrm{e}^{i(\alpha_m(t_2)-\alpha_{m'}(t_2)-\omega t_2)} \chi_{q,j}(t_2) \mel**{\phi_{m'}(t_2)}{\sigma_j}{\phi_m(t_2)}\,\mathrm{d}t_2.
\end{aligned}
\end{equation}
For a given $n$, $n'$, $m$, and $m'$, the $k$-dependent factors of this sum yields
\begin{align}
\label{eq:sum-k}
& \sum_{k} \mel**{\phi_n(0)}{\sigma_k}{\phi_{n'}(0)} \mel**{\phi_m(0)}{\sigma_k}{\phi_{m'}(0)} \nonumber\\
& =
\begin{cases}
1 & \text{if } \left\lbrace n,n',m,m' \right\rbrace = \left\lbrace \pm, \pm, \pm, \pm \right\rbrace \\
-1 & \text{if } \left\lbrace n,n',m,m' \right\rbrace = \left\lbrace \pm, \pm, \mp, \mp \right\rbrace \\
2 & \text{if } \left\lbrace n,n',m,m' \right\rbrace = \left\lbrace \pm, \mp, \mp, \pm \right\rbrace \\
0 & \text{otherwise}
\end{cases}.
\end{align}
We can use Eqs.~\eqref{eq:eigenvectors-1}, \eqref{eq:eigenvectors-2}, \eqref{eq:sum-k} as well as the fact that $\mel**{\phi_{\pm}(t)}{\sigma_{k}}{\phi_{\pm}(t)} = -\mel**{\phi_{\mp}(t)}{\sigma_{k}}{\phi_{\mp}(t)}$ to simplify Eq.~\eqref{eq:filter-function-intermediate} into
\begin{widetext}
\begin{align}
\label{eq:penultimate-FF}
    F_{q}(\omega) &= \sum_{i,j}\left(\int_{0}^{T}\mel**{\phi_+(t)}{\sigma_i}{\phi_+(t)} \chi_{q,i}(t)e^{\imath \omega t}\,\mathrm{d}t\right) \left(\int_{0}^{T}\mel**{\phi_+(t)}{\sigma_j}{\phi_+(t)} \chi_{q,j}(t)e^{-\imath \omega t}\,\mathrm{d}t\right)\nonumber\\
    &+\frac{1}{2}\left(\int_{0}^{T}\mel**{\phi_-(t)}{\sigma_i}{\phi_+(t)} \chi_{q,i}(t)e^{\imath 2\alpha(t) + \imath \omega t}\,\mathrm{d}t\right) \left(\int_{0}^{T} \mel**{\phi_+(t)}{\sigma_j}{\phi_-(t)} \chi_{q,j}(t)e^{-\imath 2\alpha(t) - \imath \omega t}\mathrm{d}t\right)\nonumber\\
    &+\frac{1}{2}\left(\int_{0}^{T}\mel**{\phi_+(t)}{\sigma_i}{\phi_-(t)} \chi_{q,i}(t)e^{-\imath 2\alpha(t) + \imath \omega t}\,\mathrm{d}t\right) \left(\int_{0}^{T} \mel**{\phi_-(t)}{\sigma_j}{\phi_+(t)} \chi_{q,j}(t)e^{\imath 2\alpha(t) - \imath \omega t}\,\mathrm{d}t\right).
\end{align}
\end{widetext}
Finally, substituting in Eqs.~\eqref{eq:eigenvectors-1} and \eqref{eq:eigenvectors-2} allows us to compactly write Eq.~\eqref{eq:penultimate-FF} in the following vectorized expression:
\begin{equation}
\label{eq:filter-function-final}
    F_q(\omega) = \norm{\int_{0}^{T} \Lambda(t) \begin{bmatrix}
        \chi_{q,X}(t)\\\chi_{q,Y}(t)\\\chi_{q,Z}(t)
    \end{bmatrix} e^{i\omega t}\,\mathrm{d}t}^2,
\end{equation}
where the entries of the matrix $\Lambda$ are given by
\begin{equation}
\begin{aligned}
\Lambda_{11} &= \cos\beta \sin\gamma,\\
\Lambda_{12} &= \sin\beta \sin\gamma,\\
\Lambda_{13} &= \cos\gamma,\\
\Lambda_{21} &= -\cos\beta\cos\gamma\cos\zeta - \sin\beta\sin\zeta,\\
\Lambda_{22} &= -\sin\beta\cos\gamma\cos\zeta + \cos\beta\sin\zeta,\\
\Lambda_{23} &= \sin\gamma\cos\zeta,\\
\Lambda_{31} &= -\cos\beta\cos\gamma\sin\zeta + \sin\beta\cos\zeta,\\
\Lambda_{32} &= -\sin\beta\cos\gamma\sin\zeta - \cos\beta\cos\zeta,\\
\Lambda_{33} &= \sin\gamma\sin\zeta.
\end{aligned}
\end{equation}
This is our main result and we show in the following sections some examples of its utility. Before we proceed, we comment on the form of Eq.~\eqref{eq:filter-function-final}. First, although the similarity between Eqs.~\eqref{eq:filter-function}-\eqref{eq:R-freq-dom} and \eqref{eq:filter-function-final} might seem to suggest that $R(t)$ and $\Lambda(t)$ are identical and we have not really simplified anything, in fact what we have done is to note that the dependence of $R(t)$ on the value of the dynamical invariant parameters evaluated at $t = 0$ does not affect the filter function value, and $\Lambda(t)$ does not carry that extraneous dependence. Second, certain error models admit an alternative interpretation for Eq.~\eqref{eq:filter-function-final}. For example, suppose we consider the dephasing and over-rotation noise models. The former can be induced by an additive shift to the qubit detuning, $\Delta(t) \rightarrow \Delta(t) + \delta_\Delta(t)$, and the latter can be induced by a multiplicative shift in the pulse amplitude, $\Omega(t) \rightarrow \Omega(t)\left(1 + \delta_\Omega(t)\right)$. The corresponding error sensitivities are $\bm{\chi}_{\Delta}(t) = \frac{1}{2}\left[0,0,1\right]^\intercal$ and $\bm{\chi}_{\Omega}(t) = \frac{1}{2}\left[\Omega\cos\varphi,\Omega\sin\varphi,0\right]^\intercal$. Substituting these expressions onto Eq.~\eqref{eq:filter-function-final} yields the following filter functions
\begin{gather}
\label{eq:FF-Delta}
    F_\Delta(\omega) = \norm{\bigintsss_{0}^{T} \frac{1}{2}\begin{bmatrix}
        \cos\gamma\\\sin\gamma\cos\zeta\\\sin\gamma\sin\zeta
    \end{bmatrix} e^{i\omega t}\,\mathrm{d}t}^2,\\
\label{eq:FF-Omega}
    F_\Omega(\omega) = \norm{\bigintsss_{0}^{T} \frac{1}{2}\begin{bmatrix}
        \dot{\zeta}\sin^2\gamma\\\dot{\zeta}\sin\gamma\cos\gamma\cos\zeta + \dot{\gamma}\sin\zeta\\\dot{\zeta}\sin\gamma\cos\gamma\sin\zeta - \dot{\gamma}\cos\zeta
    \end{bmatrix} e^{i\omega t}\,\mathrm{d}t}^2.
\end{gather}

Up to a scalar factor, the detuning filter function in Eq.~\eqref{eq:FF-Delta} can be reinterpreted as a position vector with constant speed $\dot{\vec{r}} = \left[\cos\gamma,-\sin\gamma\cos\zeta,-\sin\gamma\sin\zeta\right]$~\footnote{The sign difference in comparison with Eq.~\eqref{eq:FF-Delta} is a consequence of our choice of parameterization for the dynamical invariant eigenvectors and is irrelevant since only the magnitude of $\dot{\vec{r}}$ matters.}. If robustness at a certain noise frequency is defined by a vanishing filter function value, robustness against static detuning noise (i.e., at $\omega = 0$) is equivalent to having the position vector trace a closed three-dimensional curve whose curvature $\kappa$ is given by $\Omega$. Such a geometric interpretation has been noted previously in the literature~\cite{Zeng_2018,Zeng_2018_2,Zeng_2019,Buterakos_2021,Dong_2021,Barnes_2021,Zhuang_2022}. 

A similar observation can be made for the pulse amplitude filter function. Note that the vector in the integrand of Eq.~\eqref{eq:FF-Omega} is equivalent to $\dot{\vec{r}} \cross \ddot{\vec{r}}$. This can be rewritten as $\Omega\,\vec{b}$~\cite{Kreyszig_2013}, where $\vec{b}$ is the binormal vector corresponding to $\vec{r}$ and we have used the fact that the curvature $\kappa = \Omega$. Therefore, constructing a quantum gate that is simultaneously robust against static detuning and pulse amplitude noise is mathematically equivalent to finding a closed three-dimensional curve such that $\int_{0}^{T} \Omega(t)\vec{b}(t) \,\mathrm{d}t = \vec{0}$. As far as we know, this has not been noted before.

\section{Broadband Noise Optimization}
\label{sec:results}

We demonstrated in Sec.~\ref{sec:dynamical-invariants} that it is possible through Hamiltonian reverse engineering to analytically calculate the filter function of an arbitrary one-qubit gate in terms of the dynamical invariant parameters $\beta(t)$, $\gamma(t)$, and $\zeta(t)$ as well as the sensitivity $\chi_{q,i}(t)$. One immediate implication of this result is the possibility of filter function engineering which can be used for error suppression~\cite{Ball_2021,Bentley_2020,Baum_2021,Carvalho_2021} or quantum sensing~\cite{Norris_2018}. In the context of error suppression, we can use Eq.~\eqref{eq:filter-function-final} to define a cost function which can be minimized in spectral regions where the noise PSD is dominant. This approach allows us to target any robust one-qubit gate provided that we can find an appropriate $\gamma(t)$ and $\beta(t)$. Furthermore, this is different from previous filter function engineering results since calculating the evolution operator is no longer necessary, which helps to reduce the computational workload of the optimization framework.

We consider again as an example the case where our system is subject to detuning and pulse amplitude noise. Note that both Eqs.~\eqref{eq:FF-Delta} and \eqref{eq:FF-Omega} depend only on $\gamma$ and $\zeta$. This means that $\beta$ is a free parameter up to the boundary conditions imposed by the reverse engineering process. This extra degree of freedom can be used to impose control restrictions such as strict two-axis control. Combining Eqs.~\eqref{eq:auxiliary-equations-1}, \eqref{eq:auxiliary-equations-2}, and \eqref{eq:zeta} provides us with the reverse engineered Hamiltonian parameters in terms of the dynamical invariant parameters:
\begin{gather}
    \label{eq:Omega-IE}
    \Omega = \sqrt{\dot{\gamma}^2 + \dot{\zeta}^2\sin^2\gamma}\\
    \label{eq:varphi-IE}
    \varphi = \beta - \arctan \frac{\dot{\gamma}}{\dot{\zeta}\sin\gamma}\\
    \label{eq:Delta-IE}
    \Delta = \dot{\beta} - \dot{\zeta}\cos\gamma.
\end{gather}
For simplicity, we can set $\Delta = 0$ by solving the differential equation $\dot{\beta} = \dot{\zeta}\cos\gamma$ for $\beta$ with the boundary condition $\beta(0) = -\zeta(0)$. Thus, all properties of the output gate is determined by $\gamma$ and $\zeta$.

Restricting $\beta$ in this manner does not necessarily diminish our ability to target arbitrary one-qubit gates. In practice, a finite set of quantum gates are used to target arbitrary operations. Although we can engineer $\gamma$ and $\zeta$ to target gates directly, it is worth pointing out that many qubit implementations have access to virtual $Z$ ($\textsc{vz}$) gates~\cite{Knill_2000,Knill_2008,Johnson_2015,McKay_2017}. These zero-duration gates are essentially perfect and implemented through abrupt changes to the reference phase. We can take advantage of virtual gates by noting that any one-qubit operation can be decomposed into the product of $Z$ gates and two $X_{\frac{\pi}{2}}$~\cite{McKay_2017}: 
\begin{equation}
\label{eq:U-target-decomp}
U_{\text{target}} = Z_{\theta_1} X_\frac{\pi}{2} Z_{\theta_2} X_\frac{\pi}{2} Z_{\theta_3}.
\end{equation}
More generally, the reverse engineering method allows us to replace $X_\frac{\pi}{2}$ in the gate decomposition with $U_c(T)$. We can rewrite the engineered gate in Eq.~\eqref{eq:time-evolution-rewrite} as
\begin{align}
U_c(T) &= Z_{\beta(T)} Y_{\gamma(T)} Z_{\zeta(0)-\zeta(T)} Y_{-\gamma(0)} Z_{-\beta(0)} \nonumber\\
\label{eq:zxz-decomposition}
&= Z_{\psi_1} X_{\theta} Z_{\psi_2},
\end{align}
where
\begin{align}
\cos\left(\theta\right) &= \cos(\zeta(0)-\zeta(T)) \sin(\gamma(T))\sin(\gamma(0)) \nonumber \\
\label{eq:theta}
&\qquad + \cos(\gamma(T))\cos(\gamma(0)) ,
\end{align}
and $\psi_1$ and $\psi_2$ are angles that depend on the target gate. Setting $\theta = \frac{\pi}{2}$, we find that $X_\frac{\pi}{2} = Z_{-\psi_1} U_c(T) Z_{-\psi_2}$. This expression can be substituted onto Eq.~\eqref{eq:U-target-decomp} which yields
\begin{equation}
\label{eq:U-target-decomp-2}
U_\text{target} = Z_{\theta_1 - \psi_1} U_c(T) Z_{\theta_2 - \psi_1 - \psi_2} U_c(T) Z_{\theta_3 - \psi_2}.
\end{equation}
Since $Z$ gates are executed virtually, we only need one physical gate, $U_c(T)$ with $\theta = \frac{\pi}{2}$, to produce any one-qubit operation.

Hence, our goal is to optimize $U_c(T)$ by minimizing the following cost function:
\begin{gather}
    \text{cost} = c_1 \int_{-\infty}^{\infty} F_\Delta(\omega)S_\Delta(\omega)\,\mathrm{d}\omega + c_2 \int_{-\infty}^{\infty} F_\Omega(\omega)S_\Omega(\omega)\,\mathrm{d}\omega \nonumber\\
    + c_3 \abs{\cos\left(\zeta(0)-\zeta(T)\right)\sin\gamma(T)\sin\gamma(0) + \cos\gamma(T)\cos\gamma(0)} \nonumber\\
    + c_4 \abs{\frac{\Omega(0)}{\Omega_\text{max}}} + c_5 \abs{\frac{\Omega(T)}{\Omega_\text{max}}} + c_6 \sum_i \text{max}\left(0,\frac{\Omega(t_i)}{\Omega_\text{max}} - 1\right) \nonumber\\
    + c_7 \sum_i \text{max}\left(0,\frac{\abs{\dot{\Omega}(t_i)}}{\Omega_\text{max}/T_\text{ramp}} - 1\right).
    \label{eq:cost}
\end{gather}
The first two terms correspond to the infidelity integrals for detuning and amplitude noise with noise PSD $S_\Delta$ and $S_\Omega$, respectively. The third term is the constraint that targets $\theta = \frac{\pi}{2}$. The fourth and fifth term sets the boundary value of the pulse amplitude to zero~\footnote{These constraints are not necessary but they help with the overall experimental feasibility of the pulses we produce.}. The sixth term imposes a maximum value $\Omega_\text{max}$ on $\Omega$ by discretizing the interval $\left[0,T\right]$ and evaluating $\Omega$ at each time value. The cost penalizes any point where $\Omega(t_i) > \Omega_\text{max}$ through the function
\begin{equation}
\max(0,x) = \begin{cases}
0 & x\leq0\\
x & x>0
\end{cases}.
\end{equation}
The seventh term imposes a bound on the slope of $\Omega$. This accounts for the slew rate of the hardware that produces our control pulse. We assume a maximum rate of change of $\Omega_\text{max}/T_\text{ramp}$. Finally, $c_i$ are weighting parameters that can be adjusted to ensure that the constraints are satisfied.

We demonstrate the flexibility of our approach by considering two examples. We first consider a case where the goal is to produce a gate that acts as a stopband filter against $1/f$ detuning and pulse amplitude noise. We then consider a case where the goal is to produce a gate that is optimal in the presence of $1/f$ pulse amplitude noise and a static detuning noise. To this end, we employ deep neural networks~\cite{LeCun_2015,Schmidhuber_2015} as our optimization framework. The power of neural networks originates from their ability to represent complex ideas as a hierarchy of simpler concepts. This allows them to efficiently identify key abstract properties of a problem, which is highly coveted in tasks such as pattern recognition~\cite{Krizhevsky_2017}. It has also been proven that neural networks with sufficient neurons and layers can act as a universal function approximator~\cite{Cybenko_1989,Hornik_1991}. This is ideal for our purpose since it eliminates the nontrivial task of finding suitably parameterized ansatz function to optimize over that will yield convergent solutions. Furthermore, machine learning frameworks tend to have built-in automatic differentiation capabilities which can be utilized for gradient-based optimization.

\begin{figure}
\centering
\includegraphics[scale=.341]{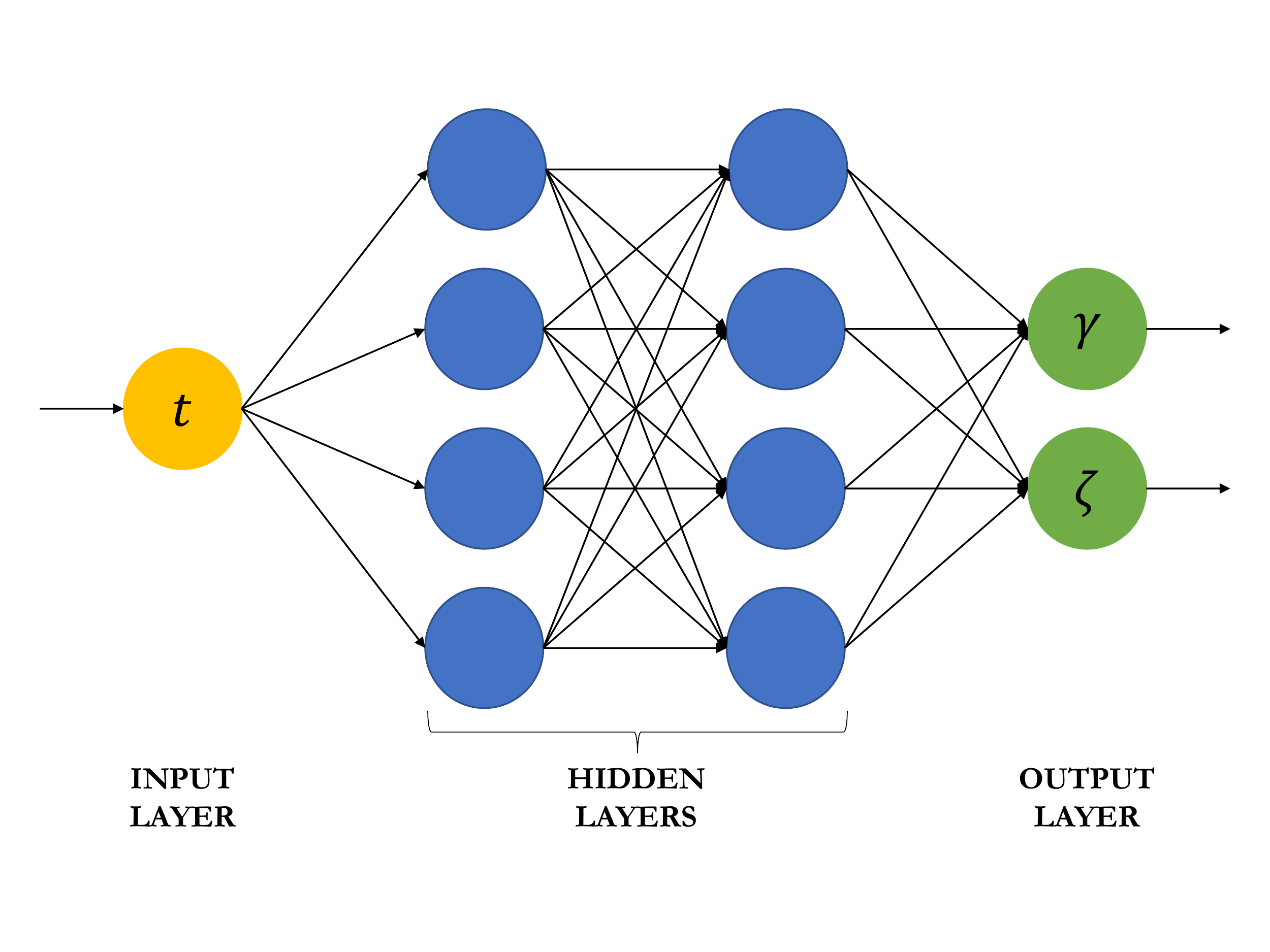}
\caption{A schematic diagram of a feedforward deep neural network with one input neuron, two output neurons, and two hidden layers with four neurons each. A neural network is deep if it has at least two hidden layers. As information flows from the input layer, each subsequent layer nonlinearly transforms incoming information and returns a value. The goal is to train the neural network so that the final output optimizes the cost. In our case, we would like to train a neural network to take time as input and return the optimized dynamical invariant parameters $\gamma$ and $\zeta$.}
\label{fig:neural-network}
\end{figure}

In particular, we use a feedforward neural network (sometimes referred to as multilayer perceptron) which is constructed using layers of interconnected computational units called neurons such that information travels only in one direction; starting with an input layer, then a series of hidden layers, and finally onto an output layer. A schematic diagram of a feedforward neural network is shown in Fig.~\ref{fig:neural-network}. Each adjacent layers act as a function that takes a vector input and produces a vector output using the following model
\begin{equation}
\label{eq:dnn-model}
\bm{x}_{i+1} = \sigma\left(W_i \bm{x}_i + \bm{b}_i\right),
\end{equation}
where $\bm{x}_i$ is the input in the $i^\text{th}$ layer, $W_i$ is a matrix that describes the neural connections between the $i^\text{th}$ and $(i+1)^\text{th}$ layer, $\bm{b}_i$ is a bias vector, and $\sigma(\cdot)$ is a nonlinear activation function such as $\text{max}(0,\cdot)$ or $\tanh(\cdot)$. Our goal is to train the neural network using optimization algorithms (e.g., ADAM~\cite{Kingma_2014}, L-BFGS~\cite{Liu_1989}, and BFGS~\cite{Fletcher_2000}) to return the optimized dynamical invariant parameters $\gamma$ and $\zeta$ on the output layer by feeding in time on the input layer. For our optimization we use a feedforward deep neural network with one input neuron, two hidden layers with 32 neurons each and a $\tanh$ activation function, and two output neurons for a total of 1186 parameters~\footnote{In feedforward neural networks, each neural connection adds one free parameter. Furthermore, each receiving neuron applies a bias parameter to incoming data. Thus, if we have a 1-3-2 network (one input neuron, one hidden layer with three neurons, and two output neurons), we have $(1*3+3) + (3*2+2) = 14$ free parameters to optimize. In our work, we used a 1-32-32-2 network that has $(1*32+32) + (32*32+32) + (32*2+2) = 1186$ free parameters.}.

\subsection{\texorpdfstring{$1/f$}{1/f} stopband filter for both detuning and pulse amplitude noise}\label{subsec:both1overf}

\begin{figure}
\centering
\includegraphics[scale=.485]{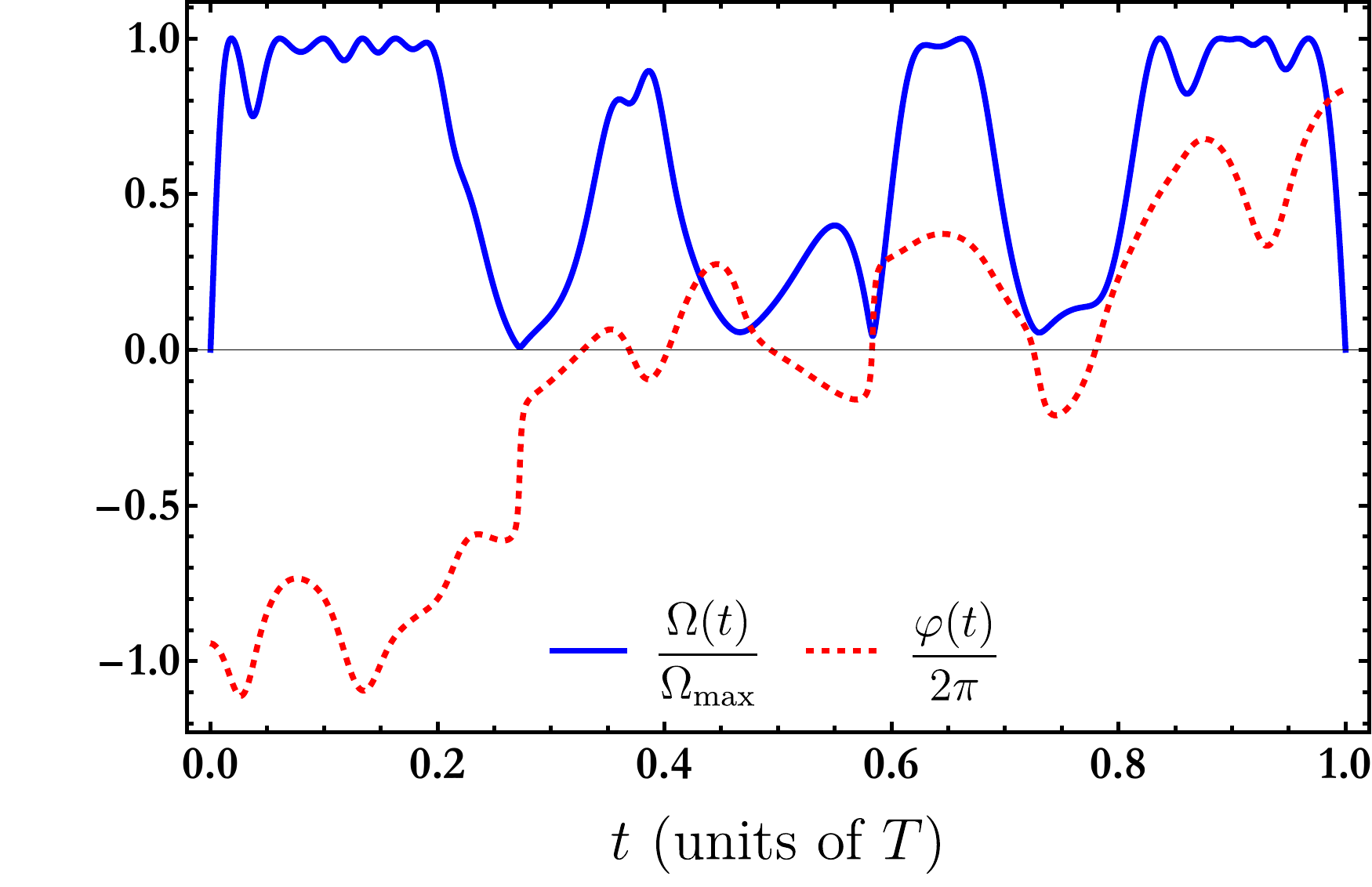}
\centering
\includegraphics[scale=.485]{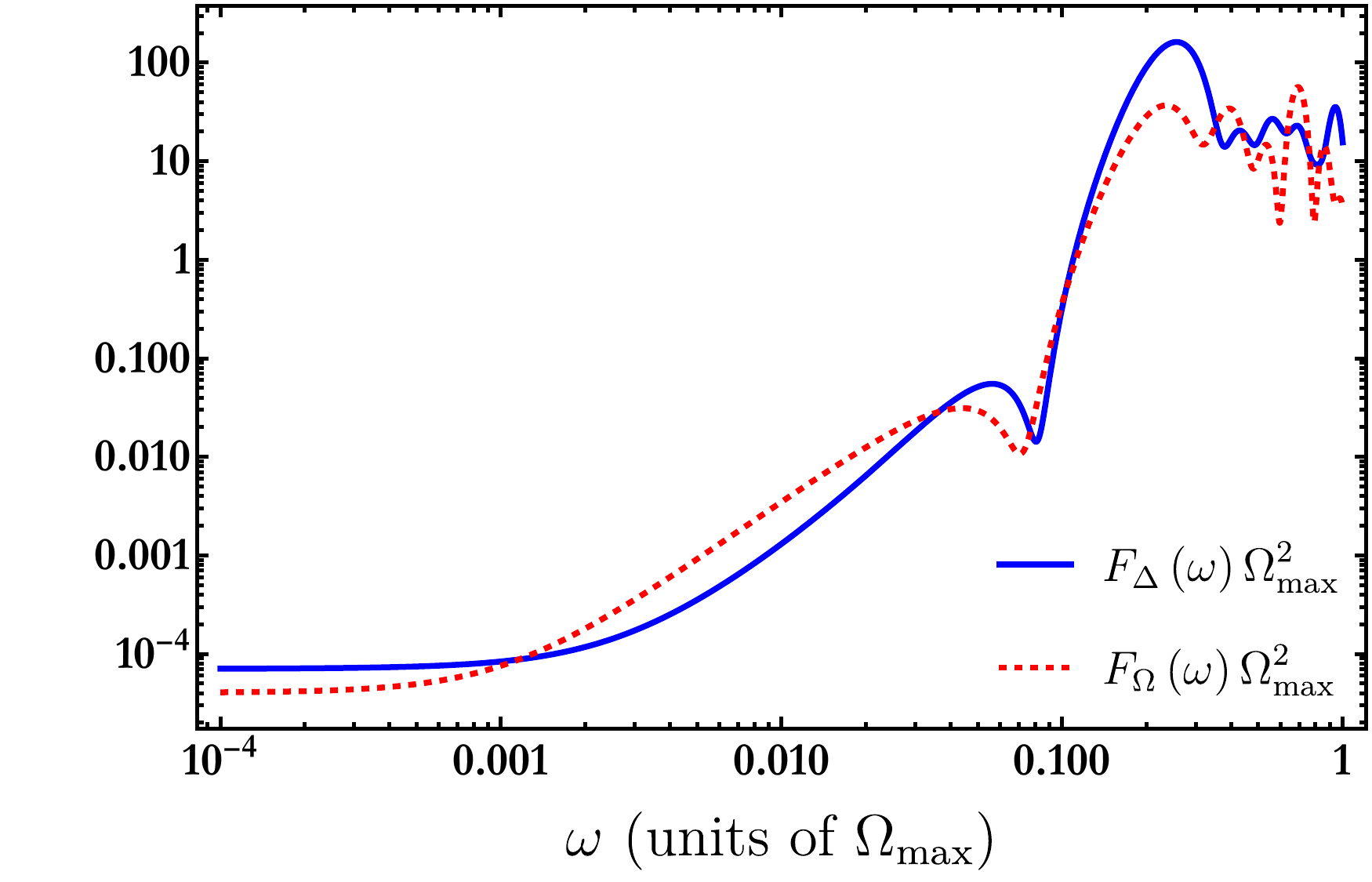}
\caption{A plot of the optimized Hamiltonian (TOP) and filter function (BOTTOM) for the case of simultaneous $1/f$ detuning and pulse amplitude noise over a finite frequency range.}
\label{fig:case1}
\end{figure}

For our first example, we consider identical noise PSD for detuning and amplitude noise:
\begin{equation}
\label{eq:PSD-1}
S_\Delta(\omega) = S_\Omega(\omega) = \begin{cases} \frac{A}{\omega} & \omega_0 \leq \abs{\omega} \leq \omega_c\\ 0 & \text{otherwise} \end{cases},
\end{equation}
where $\left[\omega_0,\omega_c\right]$ defines the frequency stopband in which we wish to suppress noise. We set $\omega_0 = 10^{-9}\Omega_\text{max}$, $\omega_c = 10^{-1}\Omega_\text{max}$, $T = 16\pi/\Omega_\text{max}$, and $T_\text{ramp} = 0.5/\Omega_\text{max}$. We present in Fig.~\ref{fig:case1} a plot of the optimized control fields and filter functions. The details of our numerical optimization scheme is provided in App.~\ref{app:method}.

We see from Fig.~\ref{fig:case1} that the control pulse we produced satisfies the imposed constraints. We compare the total infidelity of our optimized pulse with that of known pulse sequences in the literature that address either detuning noise, pulse amplitude noise, or both. We present in Table~\ref{table:DNN} a summary of these comparisons. We establish a fixed reference point by setting the noise PSD amplitude $A$ so that the naive pulse has an infidelity of $10^{-1}$. Furthermore, we also assume that the Magnus expansion converges and that the first-order filter function is sufficient to estimate the infidelity (see App.~\ref{app:derivation}). The reverse engineered gate can be related to $X_\frac{\pi}{2}$ (up to a global phase factor) by using $\psi_1 = -1.1617\pi$ and $\psi_2 = 1.7348\pi$ in Eq.~\eqref{eq:zxz-decomposition}. We find that our broadband optimized pulse yields an infidelity that is at least an order of magnitude lower than than any other pulse sequences. Specifically, the minimum improvement is roughly a factor of 27 which is a comparison with the concatenated CORPSE~\cite{Cummins_2000,Cummins_2001,Cummins_2003} and BB1~\cite{Wimperis_1994} pulse sequence (CinBB)~\cite{Bando_2013}. CinBB is designed to mitigate static additive detuning and multiplicative pulse amplitude noise simultaneously. The difference in performance between CinBB and our engineered pulse can be attributed to the fact that composite pulse sequences are generally designed to suppress static noise. Although composite pulses offer some protection against noise in the quasistatic frequency regime, their ability to suppress noise that fluctuate on the order of $\Omega_\text{max}$ is severely limited. At worst, they can even amplify the detrimental effects of such noise sources. 

Broadly speaking, suppressing noise that fluctuate at a certain frequency would require control field modulation at a higher frequency~\cite{Green_2013}. Our optimization scheme takes advantage of this fact by generating pulse shapes with reduced frequency response (as characterized by the filter function) inside the stopband. On the other hand, this also causes the optimized pulse to respond strongly to noise frequencies above $\omega_c$. In other words, the performance improvement in our optimized pulse comes at the cost of increased noise sensitivity in frequency regions beyond the indicated stopband. This behavior is typical when suppressing broadband noise and can be addressed by modifying the stopband range~\cite{Ball_2015}. We note that constraining $\dot{\Omega}$ to account for hardware limitations can prevent the optimizer from finding solutions that effectively suppress the target noise. 

\begin{figure}
\centering
\includegraphics[scale=.485]{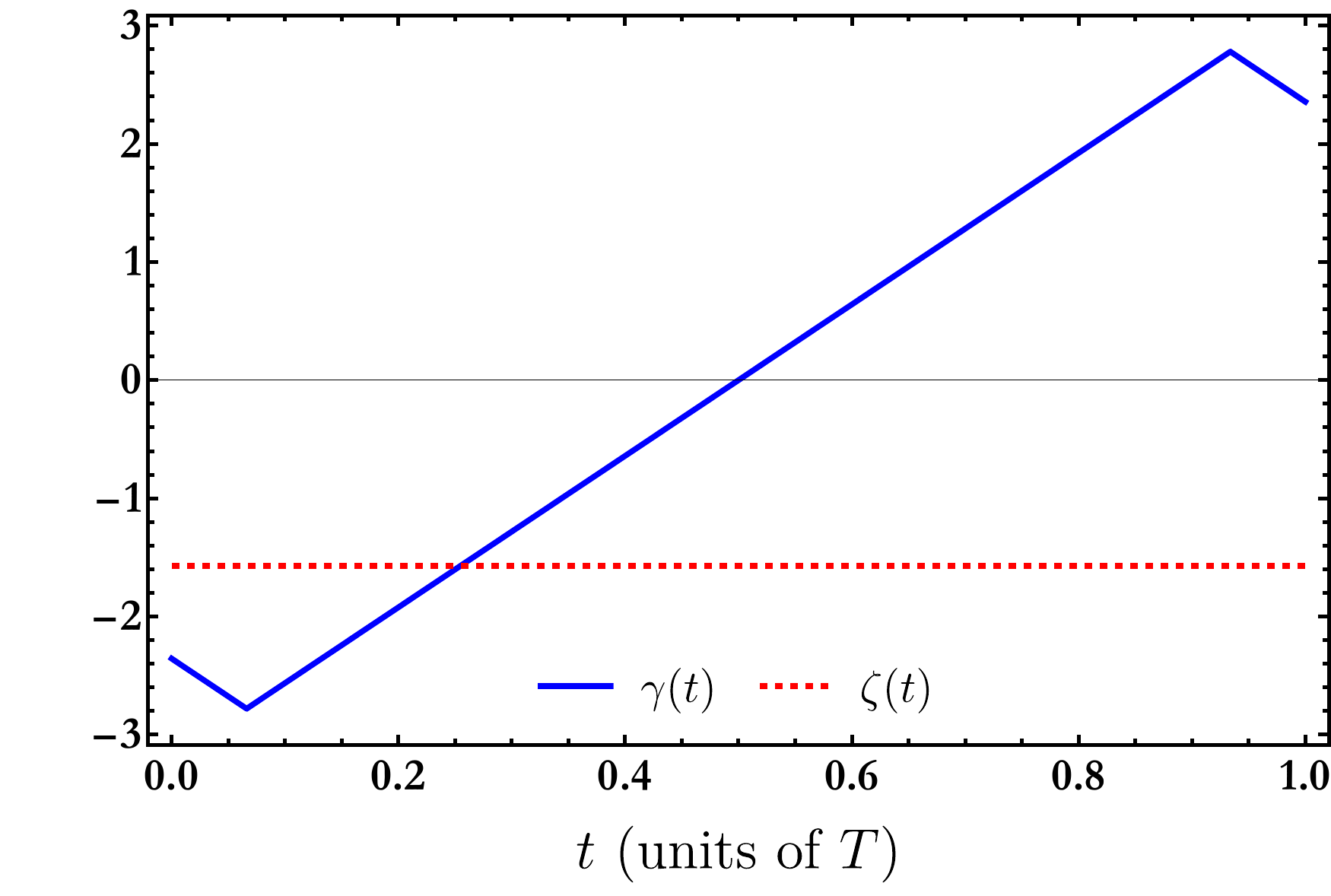}
\caption{A plot of the dynamical invariant parameters $\gamma$ and $\zeta$ for the CORPSE pulse sequence targeting $X_\frac{\pi}{2}$.}
\label{fig:symmetry}
\end{figure}

We can also investigate the effects of symmetry using our theoretical framework. We say a pulse is symmetric if $\Omega(t) = \Omega(T-t)$ and antisymmetric if $\Omega(t) = -\Omega(T-t)$. Symmetry arguments have been used in certain static noise models to analytically derive robustness conditions for the control parameters~\cite{Daems_2013,Barnes_2015,Gungordu_2019}. Since we defined noise robustness at a certain frequency by a vanishing filter function value, enforcing static noise robustness effectively turns Eq.~\eqref{eq:filter-function-final} into a vector of average integrals. If the dynamical invariant parameters $\gamma$ and $\zeta$ are symmetric or antisymmetric during the evolution (which then produces a symmetric $\Omega$), then certain choice of parameters can cause these averages to simultaneously vanish. One particular example is the CORPSE pulse sequence whose dynamical invariant parameters are shown in Fig.~\ref{fig:symmetry}. Here the antisymmetric $\gamma$ and symmetric $\zeta$ lead to robustness against static detuning noise since $F_\Delta(0) = 0$. We emphasize, however, that symmetry is not necessary to produce robust control fields. In general, there are infinitely many ways to choose $\gamma$ and $\zeta$ that lack symmetry properties but still satisfy the condition that $F_\Delta(0)$ (and/or $F_\Omega(0)$) equals zero~\footnote{For example, since the Hamiltonian control parameters and the dynamical invariant parameters are related by coupled ODEs, their correspondence is not unique. Changing the initial condition of the ODEs allows us to produce the same CORPSE pulse sequence in Fig.~\ref{fig:symmetry} using asymmetric $\gamma$ and $\zeta$.}

\subsection{Static detuning and \texorpdfstring{$1/f$}{1/f} pulse amplitude noise}\label{subsec:amp1overf}

\begin{figure}
\centering
\includegraphics[scale=.485]{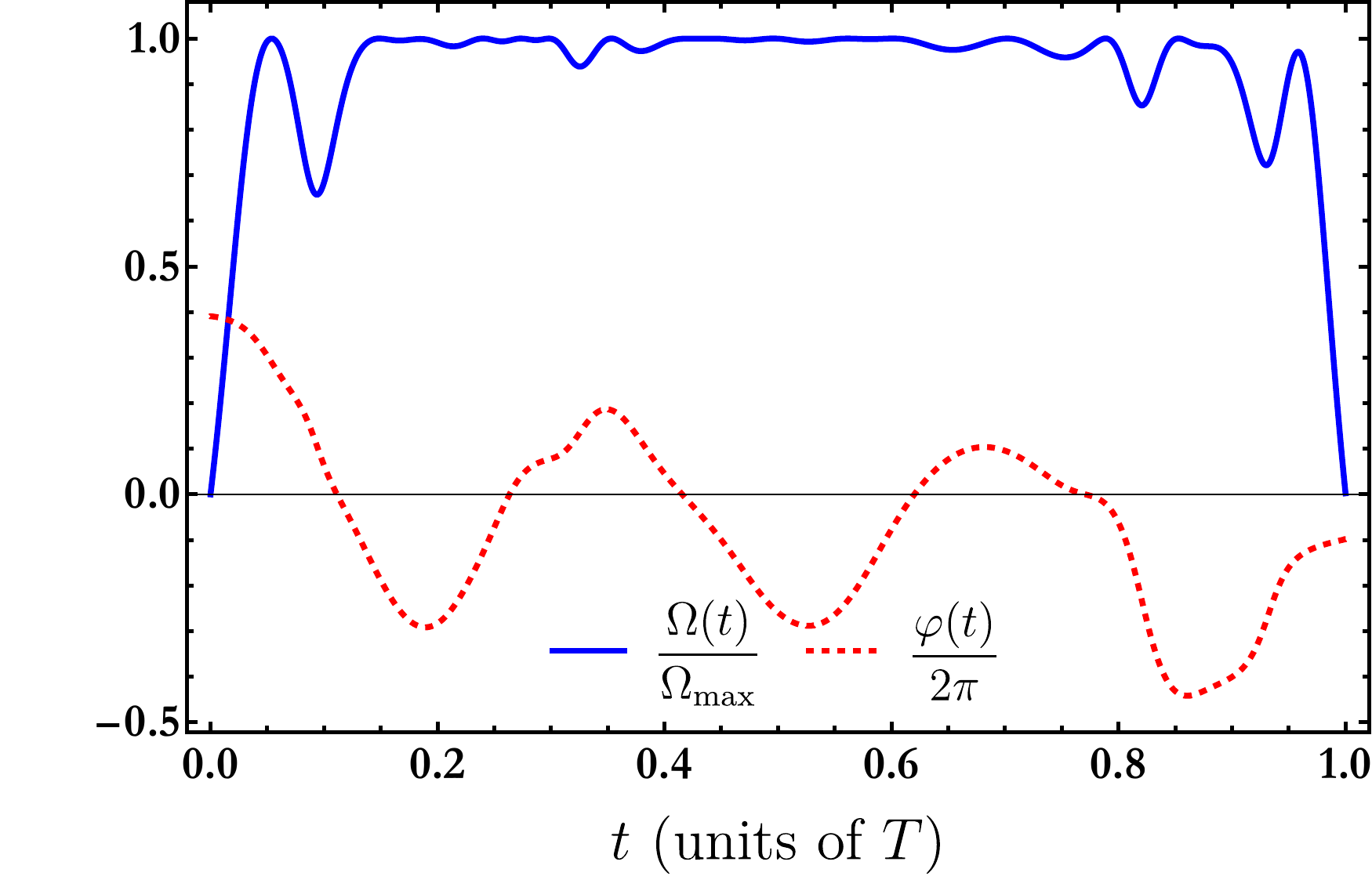}
\centering
\includegraphics[scale=.485]{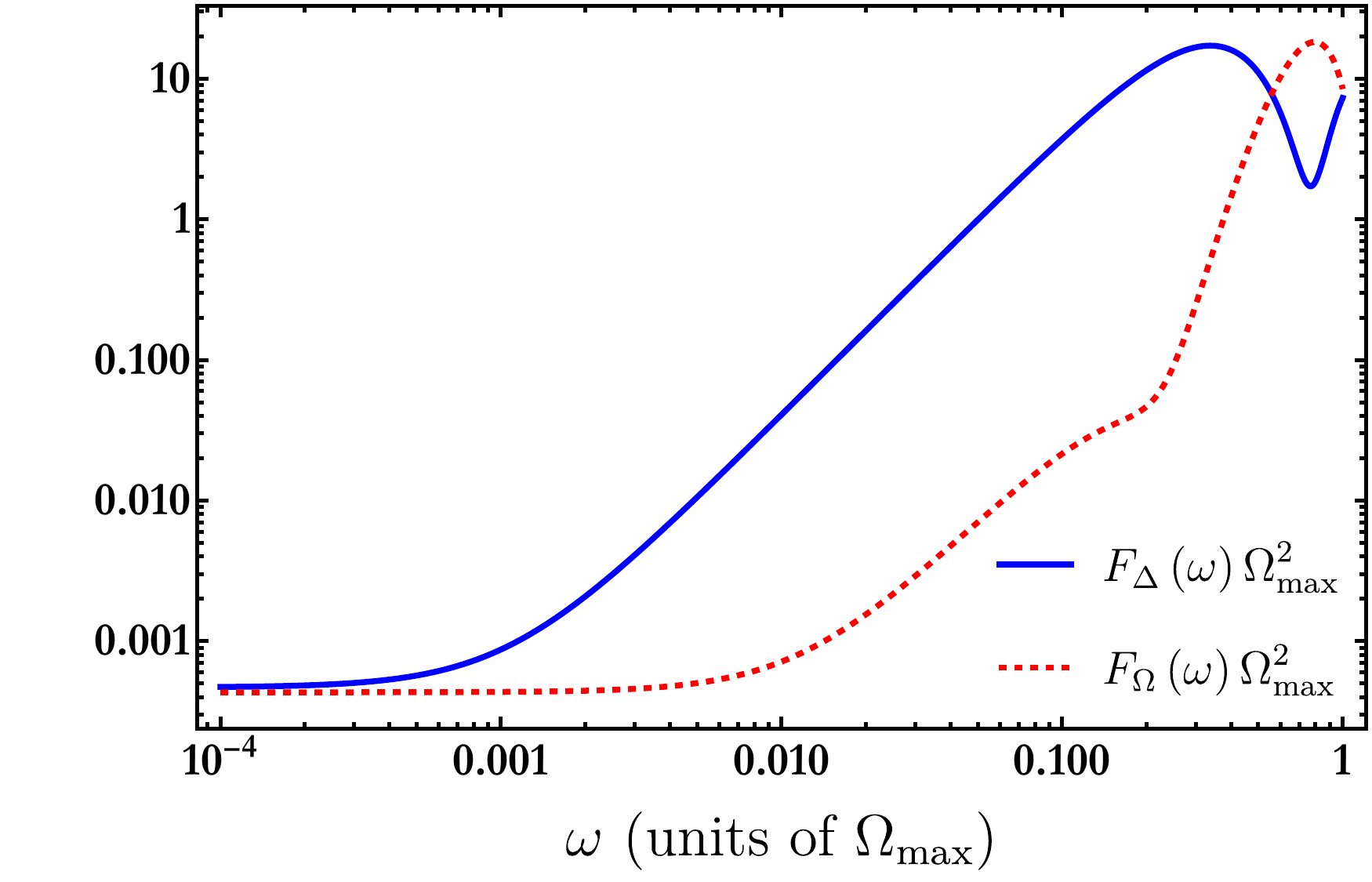}
\caption{A plot of the optimized Hamiltonian (TOP) and filter function (BOTTOM) for the case of static detuning noise and $1/f$ pulse amplitude noise. Unlike the previous example, the $1/f$ spectrum here has a $1/f^2$ tail which penalizes large filter function values in the $\omega_c \leq \abs{\omega}$ region.}
\label{fig:case2}
\end{figure}

For our second example, we consider the case where we have a static detuning noise as well as a $1/f$ pulse amplitude noise:
\begin{gather}
    S_\Delta(\omega) = 10 A \delta(\omega),\\
    \label{eq:PSD-2}
    S_\Omega(\omega) = \begin{cases} 0 & 0 \leq \abs{\omega} \leq \omega_0 \\ \frac{A}{\omega} & \omega_0 \leq \abs{\omega} \leq \omega_c\\ \frac{A \omega_c}{\omega^2} & \omega_c \leq \abs{\omega} \end{cases},
\end{gather}
where we have assumed an order of magnitude difference in the detuning and pulse amplitude noise strength. Here we set $\omega_0 = 10^{-9}\Omega_\text{max}$, $\omega_c = 10^{-1}\Omega_\text{max}$, $T = 5\pi/\Omega_\text{max}$, and $T_\text{ramp} = 0.5/\Omega_\text{max}$. We present in Fig.~\ref{fig:case2} a plot of the optimized control fields and filter functions. We again compare our optimized pulse with known pulse sequences and the results are summarized in Table~\ref{table:DNN}.

\begin{table}
\centering
\caption{A comparison of infidelities between our deep neural network (DNN) output, the naive pulse, and known composite pulse sequences. The subscript of the gate infidelity $\mathcal{I}_i$ indicates whether the case of Sec.~\ref{subsec:both1overf} or that of Sec.~\ref{subsec:amp1overf} is in consideration. The naive and composite pulses target an $X_\frac{\pi}{2}$ gate which can then be used as a building block for arbitrary one-qubit operations as shown in Eq.~\eqref{eq:U-target-decomp}. Similarly, the DNN output uses Eq.~\eqref{eq:U-target-decomp-2} to produce arbitrary one-qubit operations. We report a substantial decrease in infidelity in all cases we considered. We also indicate robustness against static detuning and/or pulse amplitude noise. Robustness is defined by a vanishing filter function at $\omega = 0$ (e.g., robustness against $\delta_\Delta$ means $F_\Delta(0) = 0$).}
\label{table:DNN}
\begin{tabular}{c|c|c|c|c}
	\hline
	\hline
	Pulse & $\mathcal{I}_\text{A}$ & $\mathcal{I}_\text{B}$ & Robust & Robust \\
	& & & to $\delta_\Delta$? & to $\delta_\Omega$? \\
	\hline
	Naive (Square) & $1.0 \times 10^{-1}$     & $1.0 \times 10^{-1}$            & No  & No\\
	Short CORPSE   & $2.1 \times 10^{-1}$     & $4.9 \times 10^{-1}$            & Yes & No\\
	BB1            & $8.3 \times 10^{-2}$     & $6.4 \times 10^{-2}$            & No  & Yes\\
	CinBB          & $7.5 \times 10^{-3}$     & $5.0 \times 10^{-2}$            & Yes & Yes\\
	CinSK          & $1.1 \times 10^{-2}$     & $1.1 \times 10^{-1}$            & Yes & Yes\\
	DNN            & $2.8 \times 10^{-4}$     & $7.4 \times 10^{-3}$            & No  & No\\
	\hline
\end{tabular}
\end{table}

The resulting gate is related to $X_\frac{\pi}{2}$ by using $\psi_1 = -1.2749\pi$ and $\psi_2 = 0.8685\pi$ in Eq.~\eqref{eq:zxz-decomposition}. Unlike the previous case, we only see a minimum improvement in infidelity by a factor of 7. In the previous example, the difference in performance is due to the fact that filter function values outside the stopband do not contribute to the infidelity. This is no longer true here due to the presence of a $1/f^2$ tail in the noise PSD that penalizes large filter function values for noise frequencies greater than $\omega_c$. Furthermore, since we cannot suppress noise that fluctuate much faster than the control fields, this effect worsens with increasing gate time. This is why we picked a smaller value of $T$ than in Sec.~\ref{subsec:both1overf}.

\section{Robustness of geometric phases}
\label{sec:robustness-of-gp}
We can also apply our result in Sec.~\ref{sec:dynamical-invariants} to explore the robustness properties of geometric quantum gates. In general, a quantum system can accumulate two types of phase~--- a dynamical phase and a geometric phase. This was first noted by Berry in the context of cyclic adiabatic evolution~\cite{Berry_1984}. In particular, it was noted that a cyclic adiabatic change in the Hamiltonian parameters produces a dynamical phase that generally depends on the duration of the evolution and a geometric phase that only depends on the geometry of the cyclic path in the Hamiltonian's parameter space. The theory of dynamical invariants can be viewed as a nonadiabatic generalization of this observation~\cite{Mostafazadeh_2001}. In particular, whereas the Hamiltonian eigenvectors form the natural basis for computing phases in the adiabatic limit, they can be replaced by dynamical invariant eigenvectors in the nonadiabatic case. Thus, analogous to Berry's result, a dynamical invariant eigenvector parameterized as in Eqs.~\eqref{eq:eigenvectors-1} or \eqref{eq:eigenvectors-2} accumulates a geometric and a dynamical phase during evolution given respectively by the following expressions:
\begin{align}
\label{eq:geometric-phase}
\alpha_{n,g}(T) &= \int_{0}^{T}\ev**{i \frac{\partial}{\partial t}}{\phi_{\pm}(t)}\,\mathrm{d}t,\\
\label{eq:dynamical-phase}
\alpha_{n,d}(T) &= -\int_{0}^{T}\ev**{H(t)}{\phi_{\pm}(t)}\,\mathrm{d}t.
\end{align}
Note that the sum of these expressions yields the Lewis-Riesenfeld phases in Eq.~\eqref{eq:lewis-riesenfeld-phase}. A geometric gate is a quantum gate for which the unitary dynamics, up to a global phase factor, is determined only by the geometric component of the total phase. This is commonly achieved by setting the integral in Eq.~\eqref{eq:dynamical-phase} to zero. Alternatively, if the qubit computational subspace is energetically degenerate, a geometric gate can still be produced even when $\alpha_{n,d}(T) \neq 0$. Since all states that belong to the subspace have the same energy, the dynamical component of the phase effectively behaves like a global phase factor. Finally, we impose the condition that $\ket{\phi_{n}(0)} = \ket{\phi_{n}(T)}$. This particular choice fixes the $U(1)$ gauge freedom on our choice of dynamical invariant eigenvectors as well as reinforce the connection between dynamical invariant theory and Berry's result.

Geometric gates are of practical interest in quantum computing due to their potential robustness against noise. Since a geometric phase depends only on the properties of its corresponding cyclic path, it is insensitive against noise that affects the speed at which the path is traversed. For this reason, geometric gates are believed to be more robust than their dynamical counterpart in certain scenarios. The validity and extent of the robustness claim remains an active area of research with many showing support for the claim~\cite{Ekert_2000, Carollo_2003, DeChiara_2003, Zhu_2005, DeChiara_2007, Wang_2007, Thomas_2011, Liang_2016, Chen_2018, Liu_2019, Chen_2020,Pachos_2001,Dong_2021,Berger_2013, Kleisler_2018, Xu_2020}. However, there are also studies that report situations in which geometric gates are not intrinsically more robust than dynamical gates~\cite{Nazir_2002,Blais_2003,Ota_2009,Zheng_2016,Dong_2021} and, in certain scenarios, their sensitivity to noise deteriorates~\cite{Solinas_2004, Carollo_2004, Zhu_2005, Dajka_2007, Johansson_2012}. It was recently shown in Ref.~\cite{Colmenar_2021} that the noise sensitivity of geometric and dynamical gates in some commonly encountered error models are generically equal in one-qubit systems with freely tuneable three-axis control. However, when control constraints are present (e.g., strict two-axis or piecewise constant control), it is possible for a particular phase type to become preferable and naturally robust.

We demonstrate in this section that a preferred phase type emerges in the case of nonadiabatic Abelian one-qubit geometric gates as a consequence of control constraints. We reiterate that a quantum gate is robust against a noise process $q$ at a particular frequency $\omega$ if $F_q(\omega) = 0$. Our reverse engineering framework is ideal for this task because it allows us to analytically compute geometric and dynamical phases in terms of the dynamical invariant parameters. Using our definition of geometric and dynamical phase in Eqs.~\eqref{eq:geometric-phase},~\eqref{eq:dynamical-phase}, the eigenvectors in Eqs.~\eqref{eq:eigenvectors-1} and \eqref{eq:eigenvectors-2}, as well as the auxiliary equations in Eqs.~\eqref{eq:auxiliary-equations-1} and \eqref{eq:auxiliary-equations-2}, we can express the geometric and dynamical phases as
\begin{align}
\label{eq:geometric-phase-2}
\alpha_{n,g}(T) &=\pm\alpha(T) \mp \int_{0}^{T} \frac{\dot{\zeta}-\dot{\beta}\cos\gamma}{2}\,\mathrm{d}t,\\
\label{eq:dynamical-phase-2}
\alpha_{n,d}(T) &= \pm \int_{0}^{T} \frac{\dot{\zeta}-\dot{\beta}\cos\gamma}{2}\,\mathrm{d}t.
\end{align}

Suppose we consider the special case of a constant detuning $\Delta$, which is a fairly common constraint in works considering geometric gates~\cite{Zhu_2002,Zhu_2005,Zhao_2017,J-Xu_2020}. We prove the following theorem for that special case by analyzing the filter function expressions that we derived:
\begin{theorem*}
Consider the one-qubit control Hamiltonian in Eq.~\eqref{eq:hamiltonian} \underline{under the constraint that $\Delta$ is constant}. Any one-qubit gate that is robust to static multiplicative amplitude noise $(\delta_\Omega)$ as well as static additive or multiplicative detuning noise $(\delta_\Delta)$ is necessarily geometric.
\end{theorem*}

\begin{proof}
Using Eq.~\eqref{eq:Delta-IE}, we can rewrite the dynamical phase integral in Eq.~\eqref{eq:dynamical-phase} as
\begin{align}
\alpha_{n,d}(T) &= \pm\int_{0}^{T} \frac{\dot{\zeta}-\dot{\beta}\cos\gamma}{2}\,\mathrm{d}t
\label{eq:reduced-dynamical-phase-integral-1}\\
&= 
\pm\int_{0}^{T} \frac{-\Delta\cos\gamma +\dot{\zeta}\sin^2\gamma}{2}\,\mathrm{d}t\nonumber\\
\label{eq:reduced-dynamical-phase-integral-2}
&= \mp\frac{\Delta}{2}\int_{0}^{T}\cos\gamma\,\mathrm{d}t \pm \frac{1}{2}\int_{0}^{T}\dot{\zeta}\sin^2\gamma\,\mathrm{d}t.
\end{align}
We begin by considering the case where there is additive detuning and multiplicative pulse amplitude noise. 
Imposing simultaneous robustness against these noise sources would require $F_\Delta(0) = F_\Omega(0) = 0$. However, we see in Eqs.~\eqref{eq:FF-Delta} and \eqref{eq:FF-Omega} that the filter function is strictly nonnegative and the only way to achieve robustness against static noise is if every integral vanishes. Specifically, robustness against static additive detuning noise requires $\int_{0}^{T}\cos \gamma\,\mathrm{d}t = 0$, while robustness against static amplitude noise requires $\int_{0}^{T}\dot{\zeta}\sin^2\gamma\,\mathrm{d}t=0$. Notice, however, that these are precisely the integral expressions in Eq.~\eqref{eq:reduced-dynamical-phase-integral-2}. Thus, simultaneous robustness against static detuning and pulse amplitude error necessarily requires the dynamical phase to vanish, i.e., the gate must be geometric. 

Next, we consider the case where there is multiplicative detuning and pulse amplitude noise. The multiplicative detuning filter function can be found using Eq.~\eqref{eq:Delta-IE} and is given by
\begin{align}
    F_{\Delta,\times}(\omega) &= \norm{\bigintsss_{0}^{T} \frac{\Delta}{2}\begin{bmatrix}
        \cos\gamma\\\sin\gamma\cos\zeta\\\sin\gamma\sin\zeta
    \end{bmatrix} e^{i\omega t}\,\mathrm{d}t}^2\nonumber\\
    &= \norm{\bigintsss_{0}^{T} \frac{\dot{\beta}-\dot{\zeta}\cos\gamma}{2}\begin{bmatrix}
        \cos\gamma\\\sin\gamma\cos\zeta\\\sin\gamma\sin\zeta
    \end{bmatrix} e^{i\omega t}\,\mathrm{d}t}^2.
\end{align}
Robustness to static noise would require $F_{\Delta,\times}(0) = 0$. We focus in particular on the first integral which can be rewritten as
\begin{equation}
    \frac{1}{2}\int_{0}^{T} \dot{\beta}\cos\gamma - \dot{\zeta} + \dot{\zeta}\sin^{2}\gamma \,\mathrm{d}t.
\end{equation}
Just like before, we note that imposing robustness against static pulse amplitude noise requires $\int_{0}^{T}\dot{\zeta}\sin^2\gamma\,\mathrm{d}t=0$ which eliminates the last term in expression above. Setting the remaining terms to zero is equivalent to setting Eq.~\eqref{eq:reduced-dynamical-phase-integral-1} to zero. Therefore, imposing simultaneous robustness against static multiplicative detuning and pulse amplitude noise necessitates a geometric gate.
\end{proof}
We make the following observations. First, this theorem is consistent with other results in the literature. It was previously noted in Refs.~\cite{Ichikawa_2012,Bando_2013} that composite pulse sequences with detuning fixed to zero that are designed to be robust against multiplicative pulse amplitude noise (and are trivially robust against multiplicative detuning noise since $\Delta = 0$) are indeed geometric quantum gates. Second, we note that in that special case of $\Delta = 0$, the first term in Eq.~\eqref{eq:reduced-dynamical-phase-integral-2} vanishes regardless of the value of the integral. In other words, if we don't require robustness to pulse amplitude noise, it \textit{is} possible to obtain dynamical gates that are robust to static detuning noise. A well-known example is the CORPSE family of composite pulses which are designed to be robust against additive detuning noise~\cite{Cummins_2000,Cummins_2001,Cummins_2003}.  Third, gates that are robust against static multiplicative pulse amplitude noise are necessarily geometric but the converse isn't true. One specific example of this is the orange-slice geometric gate presented in Ref.~\cite{Zhao_2017}. It was shown in Ref.~\cite{Colmenar_2021} that the pulse amplitude filter function in this particular case does not vanish at $\omega = 0$ despite being a geometric gate. Fourth, we note that the parallel transport condition $(\ev**{H(t)}{\phi_{\pm}(t)} = 0)$ is not necessary to achieve a robust geometric gate; the dynamical phase integral simply has to vanish at the gate time. Finally, this theorem is consistent with the results of Ref.~\cite{Colmenar_2021}. It is argued there that in the absence of control constraints, geometric and dynamical gates are generically equivalent when it comes to noise sensitivity, and preferential phase robustness can only emerge in the presence of control constraints. In this case, the constraint is considering a strictly constant $\Delta$. Removing the constraint on $\Delta$ turns $\beta$ into a free parameter. According to Eqs.~\eqref{eq:geometric-phase-2} and \eqref{eq:dynamical-phase-2}, the geometric and dynamical component of the total phase are directly  dependent on our choice of $\beta$. Thus, in the absence of constraints, we can freely tune the phase type from dynamical to geometric. Moreover, the filter functions in Eqs.~\eqref{eq:FF-Delta} and \eqref{eq:FF-Omega} are independent of $\beta$. This indicates that noise sensitivity, as quantified by the filter function, is independent of the phase type in the absence of control constraints as was also shown more generally in Ref.~\cite{Colmenar_2021}.

\section{Conclusions}
\label{sec:conclusions}

We make use of dynamical invariant theory in order to analytically reverse engineer a qubit's control Hamiltonian and calculate its corresponding filter function. This allows us to define a cost function strictly in terms of the dynamical invariant parameters which can be optimized to create filter functions with desirable properties. The primary limitation of our theory is its currently limited applicability to two-level systems, with no provision for operations on more than one qubit or correction of population leakage to higher energy levels. (The effects of \emph{virtual} transitions to higher energy levels do not pose a problem, since they can be incorporated in an effective one-qubit Hamiltonian~\cite{Schrieffer_1966}.) In those cases a generalized approach such as Ref.~\cite{Ball_2021} is preferable. However, for the specific task of constructing local rotations with robustness against high frequency noise bands, our method is a useful and efficient tool.

We demonstrate the utility of our theory by generating control pulses that are optimized to operate in the presence of broadband noise. One example we considered is creating a stopband filter for both detuning and pulse amplitude noise. We report at least an order of magnitude improvement in infidelity when our optimized pulse is compared with known composite pulse sequences that are designed to address one or both noise types. Although filter function engineering itself is not a novel concept~\cite{Ball_2021}, our approach is efficient since the reverse engineering process circumvents the need to compute the evolution operator during the optimization process. The optimizer only requires that we calculate a simple integral expression with the engineered parameters as its input. Furthermore, the engineered parameters offer adequate flexibility to simultaneously target arbitrary qubit gates while considering control parameter constraints. In principle, more complicated constraints, such as using different basis functions (Chebyshev, Walsh, Slepian, etc.), time-symmetric or antisymmetric control~\cite{Barnes_2015, Gungordu_2019, Bonesteel_2001}, or spectral-phase-only optimization~\cite{Guo_2018} to name a few, can also be incorporated into our theory. Our results can also be applied to quantum sensing where instead the goal is to maximize the filter function in a limited noise spectral bandwidth~\cite{Frey_2017,Norris_2018}.

More broadly, we used our theoretical framework to analyze the robustness of geometric gates to detuning and pulse amplitude errors. We proved a theorem for the special case of a control constraint under which one-qubit geometric gates are necessarily superior to dynamical gates. We emphasize that the robustness we report is not a generic property of geometric gates but rather a consequence of imposing control constraints.

The authors acknowledge support from the National Science Foundation under Grant No.~1915064.

\appendix

\section{Estimating gate fidelity using filter functions}
\label{app:derivation}
We now provide a more detailed derivation of the average gate infidelity provided in Eq.~\eqref{eq:avg-infidelity} which was reported in Ref. \cite{Green_2013}. We begin by writing the noisy Hamiltonian as
\begin{equation}
H(t) = H_c(t) + H_e(t),
\end{equation}
where $H_c(t)$ is the deterministic control Hamiltonian and $H_e(t)$ is the stochastic error Hamiltonian which can generally expressed as in Eq.~\eqref{eq:error-H}. By moving to the interaction frame, we can write the noisy time evolution as $U(t) = U_c(t) U_e(t)$, where $U_c$ and $U_e$ are solutions to the following Schr\"{o}dinger equations:
\begin{gather}
i \dot{U_c}(t) = H_{c}(t) U_c(t) \\
 i \dot{U_e}(t) = \left(U_c^{\dagger}(t) H_e U_c(t)\right)U_e(t).
\end{gather}
For sufficiently weak noise, we can perturbatively expand $U_e(t)$ using the Magnus expansion and write
\begin{equation}
U_e(t) \approx \exp\left[ -i \int_{0}^{T} U_c^{\dagger}(t) H_e(t) U_c(t) \,\mathrm{d}t \right].
\end{equation}
The average gate infidelity is given by
\begin{align}
\label{eq:avg-infidelity-app}
\expval{\mathcal{I}} &= \expval{1-F_\text{tr}} = \expval{1- \abs{\tr \left( U_c^{\dagger} U \right)/\tr\left( U_c^\dagger U_c \right) }^2} \nonumber\\
 &= \expval{1-\abs{\tr U_e/2}^2} \nonumber\\
 & \approx \bigg\langle \tr \int_{0}^{T} \int_{0}^{T} \left[ U_{c}^{\dagger}(t_1) H_e(t_1) U_{c}(t_1) \right]\nonumber\\
&\quad\quad\times  \left[ U_{c}^{\dagger}(t_2) H_e(t_2) U_{c}(t_2) \,\mathrm{d}t_1\,\mathrm{d}t_2\right] \bigg\rangle.
\end{align}
A sufficient condition for the convergence of the Magnus expansion can be expressed as~\cite{Moan_1999,Green_2013}
\begin{equation}
\label{eq:limitation}
\int_{0}^{T} \left(\xi_{1}^{2}(t) + \xi_{2}^{2}(t) + \xi_{3}^{2}(t) \right)^\frac{1}{2}\,\mathrm{d}t < \pi,
\end{equation}
where $\xi_i(t) = \sum_q \delta_q(t) \chi_{q,i}(t)$ as described in Eq.~\eqref{eq:error-H}. We can use the adjoint representation of $U_c(t)$ defined through
\begin{equation}
R_{ij}(t) = \frac{1}{2}\tr\left(U_{c}^{\dagger}(t)\sigma_{i}U_{c}(t)\sigma_{j}\right)
\end{equation}
and Eq.~\eqref{eq:error-H} to rewrite Eq.~\eqref{eq:avg-infidelity-app} into
\begin{align}
\label{eq:avg-infidelity-app-2}
\expval{\mathcal{I}} \approx & \sum_{q,i,j,k} \int_{0}^{T} \int_{0}^{T} \expval{\delta_q(t_1) \delta_q(t_2)} \chi_{q,i}(t_1) \chi_{q,j}(t_2)\nonumber\\
&\qquad\times R_{ik}(t_1) R_{jk}(t_2)\,\mathrm{d}t_1\,\mathrm{d}t_2.
\end{align}
We can invoke the Wiener-Khinchin theorem for a wide-sense stationary noise process to express the autocorrelation function of $\delta_q(t)$ as the Fourier transform of its PSD: $\expval{\delta_q(t_1),\delta_q(t_2)} = \frac{1}{2\pi}\int_{-\infty}^{\infty}  S_{q}(\omega)\mathrm{e}^{i\omega(t_2-t_1)}\,\mathrm{d}\omega$. If we further define
\begin{equation}
R_{q,k}(\omega) \equiv \sum_i\int_{0}^{T} \chi_{q,i}(t) R_{ik}(t) \mathrm{e}^{i\omega t}\,\mathrm{d}t,
\end{equation}
we can finally compactly write the gate infidelity as
\begin{equation}
\expval{\mathcal{I}} \approx \frac{1}{2\pi} \sum_{q} \int_{-\infty}^{\infty} S_{q}(\omega)F_{q}(\omega)\,\mathrm{d}\omega,
\end{equation}
where $F_{q}(\omega) \equiv \sum_k \left| R_{q,k}(\omega)\right|^2$. We emphasize that this expression assumes that the Magnus expansion converges which means that Eq.~\eqref{eq:limitation} is satisfied. However, this does not guarantee that contributions of higher-order filter functions to the infidelity are negligible. To this end, we can introduce the ``smallness" parameter
\begin{equation}
\xi^2 \equiv \sum_i \langle \xi_i^2(0) \rangle T^2.
\end{equation}
If $\xi^2 \ll 1$, then it can be shown that the higher-order infidelity terms can be safely neglected~\cite{Green_2013}. This consequently restricts the value of the noise PSD amplitude $A$ in Eqs.~\eqref{eq:PSD-1}--\eqref{eq:PSD-2} for which a first-order approximation is sufficient.

\section{Numerical optimization method}
\label{app:method}
We describe here the details of our numerical optimization. We used Julia's DiffEqFlux package to create a feedforward deep neural network with one input neuron, two output neurons, and two hidden layers with 32 neurons each. In principle, one hidden layer is sufficient to approximate any continuous function. However, we noticed an improvement in the optimization's convergence rate and final cost value when we added a second hidden layer. Using even deeper networks did not give any noticeable improvement and only slowed down the optimization.

Our goal is to minimize the cost given in Eq.~\eqref{eq:cost}. The infidelity integral of a noise process $q$ in the first two terms of Eq.~\eqref{eq:cost} can be expressed as
\begin{align}
    \expval{\mathcal{I}_q} &\approx \frac{1}{2\pi}\int_{-\infty}^{\infty} \int_{0}^{T} \int_{0}^{T} \left(\Lambda(t_1)\vec{\chi}_q(t_1)\right)^\intercal \Lambda(t_2) \vec{\chi}_q(t_2) \nonumber\\
    \label{eq:app-infidelity}
    &\qquad\qquad \times S_q(\omega) e^{i \omega \left(t_1 - t_2\right)} \,\mathrm{d}t_1 \,\mathrm{d}t_2 \,\mathrm{d}\omega ,
\end{align}
where $\vec{\chi} = \left[\chi_{q,X},\chi_{q,Y},\chi_{q,Z}\right]^\intercal$ is the error sensitivity vector. In the main text, the noise PSD assumes one of two nontrivial forms: $\frac{A}{\omega}$ and $\frac{A \omega_c}{\omega^2}$. We can evaluate the frequency integrals analytically which are given by
\begin{gather}
    \int_{\omega_0}^{\omega_c} \frac{A}{\omega} e^{i \omega t}\,\mathrm{d}t = 2\left(\text{Ci}\left(\omega_c t\right) - \text{Ci}\left(\omega_o t\right)\right),\\
    \int_{\omega_c}^{\infty} \frac{A \omega_c}{\omega^2} e^{i \omega t}\,\mathrm{d}t = -\pi \omega_c t + 2\cos\left(\omega_c t\right) + 2 \omega_c t\,\text{Si}\left(\omega_c t\right),
\end{gather}
where $\text{Ci}(t)$ and $\text{Si}(t)$ are the cosine and sine integral function, respectively. Let us define $g_q(t_1-t_2) = \frac{1}{2\pi}\int_{-\infty}^{\infty} S_q(\omega) e^{i\omega(t_1-t_2)}\,\mathrm{d}\omega$. This allows us to express Eq.~\eqref{eq:app-infidelity} as
\begin{equation}
\int_{0}^{T}\int_{0}^{T} g_q(t_1 - t_2) \left(\Lambda(t_1)\vec{\chi}_q(t_1)\right)^\intercal \Lambda(t_2) \vec{\chi}_q(t_2)\,\mathrm{d}t_1 \,\mathrm{d}t_2.
\end{equation}
We can approximate the integrals by converting them into a series of matrix multiplications. In particular, we can treat each time integral as an integral operator which has $g_q$ as its kernel and takes in $\bm{v}_q = \Lambda \vec{\chi}_q$ as input. Therefore, the average infidelity can be rewritten in the following bilinear form 
\begin{equation}
\label{eq:bilinear}
    \expval{\mathcal{I}_q} \approx \bm{v}_{q}^{\intercal} \mathbbm{L} \bm{v}_{q},
\end{equation}
where $\mathbbm{L}$ is a matrix that approximates the double time integral.

In our work, the cost is completely vectorized by evaluating the cost terms in evenly spaced intervals of time. The infidelity integrals are evaluated using Eq.~\eqref{eq:bilinear} while derivatives, which are used in evaluating quantities such as $\Omega$ in Eq.~\eqref{eq:Omega-IE}, are implemented using finite differences. Thus, the speed and accuracy of optimization can be controlled by choosing an appropriate level of time discretization. Finally, the relative weights are chosen to guarantee that the constraints are satisfied. The infidelity terms are equally weighted which sets $c_1 = c_2$, while the constraint terms ($c_3 - c_7$) are at least an order of magnitude larger than $c_{1,2}$.

\bibliography{filterFunctionDetuning}
\bibliographystyle{apsrev4-2}

\end{document}